\documentclass[12pt,a4paper]{article}
\usepackage{amsthm,amsfonts,amsmath,amssymb,bookmark}
\usepackage{color}

\theoremstyle{plain}
\newtheorem{theorem}{Theorem}
\newtheorem{lemma}[theorem]{Lemma}

\theoremstyle{definition}

\newtheorem{remark}[theorem]{Remark}

\theoremstyle{remark}

\def\bq{\begin{eqnarray}}
\def\eq{\end{eqnarray}}
\def\bqq{\begin{eqnarray*}}
\def\eqq{\end{eqnarray*}}
\def\nn{\nonumber}

\def\eps{\varepsilon}
\def\wto{\rightharpoonup}

\def\R{\mathbb{R}}

\def\cE {\mathcal{E}}
\def \d {{\rm d}}

\title{Ground state of the mass-critical inhomogeneous nonlinear Schr\"odinger functional}

\author{Thanh Viet Phan \\
\normalsize{Applied Analysis Research Group, Faculty of Mathematics and Statistics,} \\
\normalsize{Ton Duc Thang University, Ho Chi Minh City, Vietnam}\\
\normalsize{phanthanhviet@tdtu.edu.vn} 
}

\date{\normalsize\today}

\begin{document}

\maketitle


\begin{abstract} We study the ground state problem of the nonlinear Schr\"odinger functional with a mass-critical inhomogeneous nonlinear term. We provide the optimal condition for the existence of ground states and show that in the critical focusing regime there is a universal blow-up profile given by the unique optimizer of a Gagliardo-Nirenberg interpolation inequality. 

\bigskip

   
\end{abstract}


\section{Introduction}

We consider the mass-critical nonlinear Schr\"odinger equation
\bq\label{eq:t-NLS}
i\partial_t u(t,x)=-\Delta_x u(x) + \frac{2}{q} k(x) |u(t,x)|^{q-2} u(t,x), \quad q=2+\frac{4}{d}
\eq
with $u(t,\cdot)\in H^1(\R^d)$ and a given continuous function $k:\R^d\to \R$. 

When $k$ is a constant (the homogeneous case), \eqref{eq:t-NLS} boils down to the usual nonlinear Schr\"odinger equation studied extensively in the litt\'erature of dispersive partial differential equations (see e.g. \cite{Tao-06}). In particular, in $d=2$ dimensions it comes from the famous Gross-Pitaevskii theory describing the Bose-Einstein condensation in quantum Bose gases \cite{Gross-61,Pitaevskii-61}. 

The non-constant potential $k$ (the inhomogeneous case) corresponds to a inhomogeneous interacting effect and it arises naturally in nonlinear optics for the propagation of laser beams. Mathematically, this case is  interesting as it breaks the large group of symmetries of the homogeneous case. The study of the nonlinear Schr\"odinger equation with inhomogeneous nonlinearity was initiated by Merle  \cite{Merle-96} where he obtained a sufficient condition for the nonexistence of minimal mass blow-up solutions. On the other hand, minimal mass blow-up solutions exist if $k$ is sufficiently smooth and flat around its minima; see  Banica-Carles-Duyckaerts \cite{BanCarDuy-11} and Krieger-Schlag \cite{KriSch-09}. In $d=2$ dimensions, the full classification of minimal mass blow-up solutions in the inhomogeneous case was solved by Raphael-Szeftel \cite{RapSze-11}.

In the present paper, we are interested in the ground state solution of \eqref{eq:t-NLS}. To be precise, we will study the variational problem 
\bq \label{eq:E}
E_k= \inf \left\{ \cE_k(u): \int_{\R^d} |u|^2=1 \right\}.
\eq
associated to the nonlinear Schr\"odinger functional 
$$
\cE_k(u)= \int_{\R^d} |\nabla u(x)|^2 \d x +  \int_{\R^d} k(x) |u(x)|^{2+4/d} \d x, \quad u \in H^1(\R^d)
$$
By the standard techniques from calculus of variations, any minimizer $u_0$ of $E_k$ in \eqref{eq:E} is a solution to the stationary nonlinear Schr\"odinger equation
\bq \label{eq:u0}
-\Delta u(x) + \frac{2}{q} k(x) |u(x)|^{q-2} u(x) = \mu u(x)
\eq
with a constant $\mu\in \R$ (which is the Lagrange multiplier associated to the mass constraint $\|u\|_{L^2}=1$). Consequently, 
$$
u(t,x)=e^{-it\mu} u_0(x)
$$
is a solitary plane-wave solution to the time-dependent problem \eqref{eq:t-NLS}. 

Similarly to time-dependent problem studied in \cite{Merle-96,BanCarDuy-11,KriSch-09,RapSze-11}, a critical feature of the ground state problem \eqref{eq:E} appears when $-\inf_x k(x)$ crosses the threshold $a^*$ which is the optimal constant in the Gagliardo-Nirenberg interpolation inequality:
\bq 
\label{eq:GN} 
\left( \int_{\R^d} |\nabla u(x)|^2 \d x \right) \left( \int_{ \R^d} |u(x)|^2 \d x \right) ^{2/d} \ge a^* \int_{\R^d} |u(x)|^{2+4/d}  \d x, \quad \forall u\in H^1(\R^d).
\eq
This inequality has been well studied in \cite{GidNiNir-81,Weinstein-83,MclSer-87,Kwong-89}. It is known that \eqref{eq:GN} has a unique optimizer $Q$ up to translations and dilations. In fact, $Q$ is the unique radial positive solution to the equation  
\bq 
\label{eq:Q}
-\Delta Q + Q - Q^{1+4/d}=0.
\eq 

Our first result concerns the existence and nonexistence of minimizers of the variational problem $E_k$ in \eqref{eq:E}. 
\begin{theorem} [Existence and nonexistence of minimizers] \label{thm:existence}Assume that $k\in C(\R^d)$. 

\begin{itemize}

\item[(i)]  (Subcritical case: existence) If $\inf  k>-a^*$
and 
\bq \label{eq:k-growth-fast}
\int_{\R^d} \frac{1}{(k(x)+a^*)^{d/2}}\d x <\infty
\eq
then $E_k>0$ and it has  a minimizer. 

\item[(ii)] (Subcritical case: nonexistence) If $\inf  k>-a^*$ and 
\bq  \label{eq:k-growth-slow}
\limsup_{|x|\to \infty} \frac{k(x)}{|x|^2}<\infty,
\eq
then $E_k=0$ and it has no minimizer. 

\item[(iii)] (Critical case) If $\inf k= - a^* = k(x_0)$ for some $x_0\in \R^d$ and 
\bq \label{eq:local-degenerate}
\quad \lim_{x\to x_0} \frac{k(x)-k(x_0)}{|x-x_0|^2}=0,
\eq 
then $E_k=0$ and it has no minimizer except the case $k\equiv - a^*$. 

\item[(iv)] (Supercritical case)  If $\inf k<-a^*$, then $E_k=-\infty$. 

\end{itemize}

\end{theorem}

\begin{remark} In the subcritical case $\inf k>-a^*$, it is remarkable that the growth of $k(x)$ as $|x|\to \infty$ really matters the existence of minimizers. 
Note that in the integrability condition \eqref{eq:k-growth-fast} holds if 
$$
\liminf_{|x|\to \infty} \frac{k(x)}{|x|^{2+\eps}} >0,\quad \text{for some }\eps>0.
$$
Thus the existence condition \eqref{eq:k-growth-fast} in (i) and the nonexistence condition \eqref{eq:k-growth-slow} in (ii) are  mostly the complement to each other. 
 \end{remark}

\begin{remark} In the critical case $\inf k=-a^*$, the condition \eqref{eq:local-degenerate} means that $k$ is flat enough around its minimum point $x_0$. If $k\in C^2(\R^d)$, then \eqref{eq:local-degenerate} is equivalent to the degeneracy condition
$$\nabla^2 k(x_0)=0.$$
On the other hand, the opposite condition to  \eqref{eq:local-degenerate} that  
\bq \label{eq:local-non-degenerate}
\quad \lim_{x\to x_0} \frac{k(x)-k(x_0)}{|x-x_0|^{2-\eps}}>0,  \quad \text{for some }\eps>0,
\eq
was assumed by Merle \cite{Merle-96} when he proved the nonexistence of minimal mass blow-up solutions for the time-dependent problem. In fact,  \eqref{eq:local-non-degenerate} implies the local integrability of $(k(x)+a^*)^{-d/2}$ (note that if $(k(x)+a^*)^{-d/2}$ is integrable, then by following the proof of Theorem \ref{thm:existence} (i)  we can prove that $E_k>0$; see Remark \ref{rmk:local-int}). However, \eqref{eq:local-non-degenerate} never happens if $k\in C^2$. From our analysis, the case 
\bq \label{eq:local-non-degenerate-aaa}k\in C^2, \quad \nabla^2 k(x_0) \ne 0 
\eq
is still missing, and it is indeed related to an open question in \cite[Remark after Prop. 5.4, page 76]{Merle-96}. The difficult case \eqref{eq:local-non-degenerate-aaa} has been studied by Raphael-Szeftel \cite{RapSze-11} in the context of minimal mass blow-up solutions {in $\R^2$}, but it is not clear to us how to transfer their techniques to the ground state problem in the present paper.
\end{remark} 

Next, we concentrate on the existence case (i) in Theorem \ref{thm:existence}, and analyze the blow-up behavior when $\inf k$ tends to $-a^*$. To make the analysis rigorous, we need to impose some explicit behavior of $k$ around its minima.  

{\bf Assumption for the blow-up result.} For the following blow-up theorem, we will assume that 
$$k(x)=K(x)-a$$
with $K:\R^d \to \R$ a fixed function satisfying: 
\begin{itemize}
\item [(i)] $\inf K=0$ and $K$ has finite minima $\{x_j\}_{j=1}^J$;

\item[(ii)] For any $j$, there exists $p_j>0$ such that 
$$
\lim_{x\to x_j} \frac{K(x)}{|x-x_j|^{p_j}} =\lambda_j >0. 
$$

\item[(iii)]  $K^{-d/2}$ is integrable away from $\{x_j\}_{j=1}^J$, namely for any $R>0$, 
$$
\int_{\min_j |x-x_j|>R} K(x)^{-d/2} \d x <\infty. 
$$

\end{itemize}
Let us denote $p=\max\{p_1,...,p_J\}$, $\lambda=\min\{\lambda_j: p_j=p\}$ and $Z=\{z_j: p_j=p, \lambda_j=\lambda\}$ (the set of flattest minima of $K$).

\begin{theorem}[Blow-up profile] \label{thm:blowup} We consider the variational problem \eqref{eq:E} with $k(x)=K(x)-a$, where $K$ satisfies the above conditions with $p>2$. Let $Q_0=Q/\|Q\|_{L^2}$ with $Q$ be the unique positive radial solution to \eqref{eq:Q}. Then we have
\begin{align} \label{eq:Energy-cv}
\lim_{a\uparrow a^*} \frac{E_k}{(a^*-a)^{1-2/p}} &=  \inf_{\xi>0} \left[ \xi^2 \int_{\R^d} |Q_0|^{2+4/d} +  \xi^{2-p} \lambda \int_{\R^d} |x|^p  |Q_0|^{2+4/d} \right] \\
&=\left( \frac{\lambda p}{2} \int_{\R^d} |x|^p |Q_0|^{2+4/d} \right)^{2/p} \left( \frac{p}{p-2} \int_{\R^d} |Q_0|^{2+4/d}\right)^{1-2/p}. \nn
\end{align}
Moreover, if $u_a$ is a minimizer for $E_k$, then for any sequence $a_n\uparrow a^*$, there exist a subsequence $a_{n_\ell}\uparrow a^*$ and an element $z\in Z$ (the set of flattest minima of $k$) such that up to a phase
\bq \label{eq:GS-cv}
\lim_{a_{n_\ell}\uparrow a^*} (a^*-a_{n_\ell})^{d/(2p)} u_{a_{n_\ell}} (z + (a^*-a_{n_\ell})^{1/p} x) \to {b^{d/2} Q_0(bx)} 
\eq
strongly in $H^1(\R^d)$, where $b$ is the optimizer for the right side of \eqref{eq:Energy-cv}:
$$
b= \left( \frac{(p-2)\lambda \int_{\R^d} |x|^p  |Q_0|^{2+4/d} }{2\int_{\R^d} |Q_0|^{2+4/d}} \right)^{1/p}.
$$ 
Moreover, if $Z$ has a unique element, then \eqref{eq:GS-cv} holds true for the whole family $a\uparrow a^*$. 
\end{theorem}

This result is obtained by a concentration argument, inspired from the paper of Guo-Seiringer \cite{GuoSei-14} who studied the blow-up profile of the Bose-Einstein condensation in 2D with the homogeneous nonlinearity ($k=const$) and a trapping potential $V(x)=|x|^s$ with $s>0$ (see also \cite{Phan-17a} for a related result with attractive external potentials). Here our main task is to deal with the inhomogeneous nonlinearity, which makes the analysis both complicated and interesting in several places. 

In the following we will prove Theorem \ref{thm:existence} in Section \ref{sec:existence} and prove Theorem \ref{thm:blowup} in Section \ref{sec:blowup}.  

\section{Existence and nonexistence of minimizers} \label{sec:existence}

\begin{proof}[Proof of Theorem \ref{thm:existence}] (i) By the Gagliardo-Nirenberg inequality \eqref{eq:GN}, for all $u\in H^1(\R^d)$ with $\|u\|_{L^2}=1$ we have
\begin{align*}
\cE_k(u) &= \int_{\R^d} |\nabla u(x)|^2 + \int_{\R^d} k(x) |u(x)|^{2+4/d} \d x \\
&\ge \int_{\R^d} |\nabla u(x)|^2 + \min\{\inf k,0\} \int_{\R^d} |u(x)|^{2+4/d} \d x\\
&\ge \left(1+\frac{\min\{\inf k,0\}}{a^*} \right) \int_{\R^d} |\nabla u|^2 \ge 0.
\end{align*}
Since $\inf k>- a^*$, we deduce that $E_k\ge 0$. Moreover, if $\{u_n\}$ is a minimizing sequence for $E_k$, then $\{u_n\}$ is bounded in $H^1(\R^d)$. By the Banach-Alaoglu theorem, up to a subsequence, we can assume that $u_n\wto u$ weakly in $H^1(\R^d)$. 

Let us prove that $u_n\to u$ strongly in $L^2(\R^d)$. First, since $u_n\wto u$ weakly in $H^1(\R^d)$, Sobolev's embedding theorem implies the local convergence 
\bq \label{eq:ext-0}
\chi(|x|\le R) u_n \to \chi(|x|\le R) u \quad \text{strongly in $L^2(\R^d)$ for every $R>0$.}
\eq
Here $\chi_A$ is the characteristic function of the set $A\subset \R^d$. On the other hand, by the Gagliardo-Nirenberg inequality \eqref{eq:GN} again we have
\begin{align*}
\cE_k(u_n) \ge \int_{\R^d} |\nabla u_n|^2 + \int_{\R^d} k(x) |u_n|^{2+4/d} \ge  \int_{\R^d} (k(x)+a^*) |u_n(x)|^{2+4/d}. 
\end{align*}
Combining this with Holder's inequality we find that 
\begin{align*}
\int_{|x|>R} |u_n|^2 &\le  \left( \int_{|x|>R} (k(x)+a^*) |u_n(x)|^{2+4/d} \right)^{d/(d+2)}  \left( \int_{|x|>R} \frac{1}{(k(x)+a^*)^{d/2}} \right)^{2/(d+2)} \\
&\le (\cE_k(u_n))^{d/(d+2)}  \left( \int_{|x|>R} \frac{1}{(k(x)+a^*)^{d/2}} \right)^{2/(d+2)}.
\end{align*}
Since $\cE_k(u_n)$ is bounded uniformly in $n$ (as it converges to $E_k$) and $(k(x)+a^*)^{-d/2} \in L^1(\R^d)$ by Assumption \eqref{eq:k-growth-fast}, we obtain the uniform convergence
\bq \label{eq:ext-1}
\sup_n \int_{|x|>R} |u_n|^2 \le C \left( \int_{|x|>R} \frac{1}{(k(x)+a^*)^{d/2}} \right)^{2/(d+2)} \to 0
\eq 
as $R\to \infty$ by Lebesgue Dominated Convergence Theorem.  By the triangle inequality we can decompose  
\begin{align*}
\|u_n - u\|_{L^2(\R^d)} \le \| \chi(|x|\le R) (u_n-u) \|_{L^2(\R^d)}+ \|\chi(|x|>R) u_n \|_{L^2(\R^d)} + \|\chi(|x|>R) u\|_{L^2(\R^d)}
\end{align*}
Taking $n\to \infty$ and using \eqref{eq:ext-0} we get
\begin{align*}
\limsup_{n\to \infty} \|u_n - u\|_{L^2(\R^d)} &\le \limsup_{n\to \infty} \| \chi(|x|\le R) (u_n-u) \|_{L^2(\R^d)}  \\
&\quad + \limsup_{n\to \infty} \|\chi(|x|>R) u_n \|_{L^2(\R^d)} + \|\chi(|x|>R) u\|_{L^2(\R^d)} \\
&\le  \sup_{n} \|\chi(|x|>R) u_n \|_{L^2(\R^d)} + \|\chi(|x|>R) u\|_{L^2(\R^d)}
\end{align*}
for all $R>0$. Since the left side is independent of $R$, we can take $R\to \infty$ on the right side and conclude that
\begin{align*}
\limsup_{n\to \infty} \|u_n - u\|_{L^2(\R^d)} &\le \lim_{R\to \infty} \Big( \sup_{n} \|\chi(|x|>R) u_n \|_{L^2(\R^d)} + \|\chi(|x|>R) u\|_{L^2(\R^d)} \Big) =0.
\end{align*}
Here we have used   \eqref{eq:ext-1} for $\|\chi(|x|>R) u_n\|_{L^2}$ and Lebesgue Dominated Convergence Theorem for $\|\chi(|x|>R) u\|_{L^2}$. Thus $u_n\to u$ strongly in $L^2(\R^d)$ as $n\to \infty$. 

Consequently, $\|u\|_{L^2}=1$ since all $u_n$'s are normalized.  Next, to deduce that $u$ is a minimizer, it remains to prove that
$$
\liminf_{n\to \infty} \cE_k(u_n) \ge \cE_k(u).
$$

Since $u_n \wto u$ weakly in $H^1(\R^d)$ and $u_n\to u$ strongly in $L^2(\R^d)$, by interpolation we deduce that $u_n\to u$ strongly in $L^{p}(\R^d)$ for all $2\le p<2^*$, where  $2^*$ is the critical power in  Sobolev's embedding theorem, i.e. $2^*=2d/(d-2)$ if $d\ge 3$ and $2^*=\infty$ if $d\le 2$. In particular, we have $u_n\to u$ strongly in $L^{2+4/d}(\R^d)$. Also, by Sobolev's embedding theorem, up to a subsequence we can assume that $u_n(x)\to u(x)$ for a.e. $x\in \R^d$. Thus by Fatou's lemma and the fact that $k(x)+a^*\ge 0$, we have
$$
\liminf_{n\to \infty} \int_{\R^d} (k(x)+a^*) |u_n(x)|^{2+4/d} \d x \ge  \int_{\R^d}  (k(x)+a^*) |u(x)|^{2+4/d} \d x.
$$
Combining this with the strong convergence $u_n\to u$ in $L^{2+4/d}(\R^d)$, we deduce that
$$
\liminf_{n\to \infty} \int_{\R^d} k(x) |u_n(x)|^{2+4/d} \d x \ge  \int_{\R^d}  k(x) |u(x)|^{2+4/d} \d x.
$$
Finally, since $u_n \wto u$ weakly in $H^1(\R^d)$, we have by Fatou's lemma again (for the weak convergence in $L^2$)
$$
\liminf_{n\to \infty} \int_{\R^d}  |\nabla u_n(x)|^{2} \d x \ge  \int_{\R^d}   |\nabla u(x)|^{2} \d x.
$$
The latter two estimates show that
$$
\liminf_{n\to \infty} \cE_k(u_n) \ge \cE_k(u),
$$
which implies that $u$ is a minimizer for $E_k$. 

\bigskip

(ii) As in (i), since $\inf k >-a^*$ we have  
$$
\cE_k(u) \ge \left(1+\frac{\min\{\inf k,0\}}{a^*} \right) \int_{\R^d} |\nabla u|^2 >0 
$$
for all $u\in H^1(\R^d)$ with $\|u\|_{L^2}=1$. Thus $E_k\ge 0$, and if we can prove, under Assumption \eqref{eq:k-growth-slow}, that $E_k=0$, then clearly $E_k$ has no mimimizer. 

Let us prove the upper bound $E_k\le 0$ using the variational principle with a suitable trial function $u$. Under Assumption \eqref{eq:k-growth-slow}, there exists a constant $C>0$ such that  
$$k(x)\le C(|x|^2+1)\quad \text{for all }x\in \R^d.$$
Therefore, by the variational principle
$$
E_k \le \cE_k(u) \le \int_{\R^d} |\nabla u(x)|^2 \d x + C\int_{\R^d}(|x|^2+1)|u(x)|^{2+4/d} \d x
$$
for all $u\in H^1(\R^d)$ with $\|u\|_{L^2}=1$. Replacing $u$ by $u_\ell(x)=\ell^{d/2} u(\ell x)$, which satisfies the normalized condition $\|u_\ell\|_{L^2}=\|u\|_{L^2}=1$,  we obtain by changing of variables 
\begin{align*}
E_k &\le \cE_k(u_\ell) \le \int_{\R^d} |\nabla u_\ell (x)|^2 \d x + C\int_{\R^d}(|x|^2+1)|u_\ell(x)|^{2+4/d} \d x\\
&\le \ell^2 \int_{\R^d} |\nabla u(x)|^2 \d x + C\int_{\R^d}(|x|^2+\ell^{2})|u(x)|^{2+4/d} \d x
\end{align*}
for all $\ell>0$. Taking $\ell\to 0$ we deduce that
$$
E_k  \le C\int_{\R^d} |x|^2 |u(x)|^{2+4/d} \d x
$$
for all $u\in H^1(\R^d)$ with $\|u\|_{L^2}=1$. Equivalently, we have
$$
E_k  \le C \inf_{\varphi \in H^1(\R^d), \varphi \not\equiv 0} \frac{\int_{\R^d} |x|^2 |\varphi(x)|^{2+4/d} \d x}{\left( \int_{\R^d} |\varphi(x)|^2 \d x\right)^{1+2/d}}. 
$$
Choosing 
$$ \varphi(x)= \frac{\chi(|x|\le R)}{(|x|+1)^{d/2}},$$
we find that
$$|x|^2 |\varphi(x)|^{2+4/d} \le |\varphi(x)|^{2}$$
and hence
$$
E_k \le \frac{C}{\left( \int_{|x|\le R} (|x|+1)^{-d} \d x\right)^{2/d}}
$$
for all $R>0$. We conclude that
$$
E_k\le \lim_{R\to \infty } \frac{C}{\left( \int_{|x|\le R} (|x|+1)^{-d} \d x\right)^{2/d}}
 =0
$$
since $(|x|+1)^{-d}$ is not integrable in $\R^d$. Thus $E_k=0$ but it has no minimizer. 

\bigskip

(iii) Now assume that $\inf k=-a^*$. Then by the Gagliardo-Nirenberg inequality \eqref{eq:GN} we have
$$
\cE_k(u)\ge \int_{\R^d} |\nabla u|^2 - a^* \int_{\R^d} |u|^{2+4/d} \ge 0
$$
for all $u\in H^1(\R^d)$ with $\|u\|_{L^2}=1$. Thus $E_k\ge 0$. 

Next, from \eqref{eq:local-degenerate}, for any $\eps>0$ there exists $r=r_\eps>0$ such that 
$$k(x)\le -a^*+ \eps |x-x_0|^2, \quad \forall x\in B(x_0,r).$$ 
By the variational principle, for all $u\in H^1(\R^d)$ supported on $B(0,r)$ such that $\|u\|_{L^2}=1$, we have
\begin{align} \label{eq:iii-0}
E_k\le \cE_k(u(.-x_0)) &= \int_{\R^d} |\nabla u(x-x_0)|^2 \d x + \int_{\R^d} k(x)  |u(x-x_0)|^{2+4/d} \d x\nn\\
&= \int_{\R^d} |\nabla u(x)|^2 \d x + \int_{\R^d} k(x+x_0)  |u(x)|^{2+4/d} \d x \nn\\
&\le \int_{\R^d} |\nabla u(x)|^2 \d x + \int_{\R^d}  (\eps |x|^2-a^*)  |u(x)|^{2+4/d} \d x. 
\end{align}
We will use \eqref{eq:iii-0} with suitable trial functions $u$. 

Let $Q_0=Q/\|Q\|_{L^2}$ be the (normalized) optimizer of the Gagliardo-Nirenberg inequality  \eqref{eq:GN}, i.e. 
\bq \label{eq:GS-Q0}
\int_{\R^d} |\nabla Q_0|^2 - a^* \int_{\R^d} |Q_0|^{2+4/d}=0. 
\eq
{Take a smooth function $ w:\R^d \to [0,1]$ such that $w(x)=1$ for $|x|\le r/2$ and $w (x)=0$ if $|x|\ge r$}. For any $\ell>0$, denote
$$
v_\ell (x) = \ell^{d/2}Q_0(\ell x) w(x).
$$
Then $v_\ell$ is supported on $B(0,r)$ and $\|v_\ell\|_{L^2}\le \|Q_0\|_{L^2}=1$. Moreover, since both $Q_0$ and $|\nabla Q_0|$ are exponentially decay (see \cite[Proposition 4.1]{GidNiNir-81}), we have 
\begin{align*}
\int_{\R^d} |v_\ell(x)|^2 \d x&=  \int_{\R^d} | Q_0|^2 + o(1)_{\ell\to \infty} = 1 + o(1)_{\ell\to \infty} ,\\
\int_{\R^d} |\nabla v_\ell(x)|^2 \d x&= \ell^2 \int_{\R^d} |\nabla Q_0|^2 + o(1)_{\ell\to \infty},\\
\int_{\R^d} |v_\ell(x)|^{2+4/d} \d x &=\ell^2 \int_{\R^d} | Q_0|^{2+4/d}+ o(1)_{\ell\to \infty}\\
\int_{\R^d} |x|^2 |v_\ell(x)|^{2+4/d} \d x & \le \int_{\R^d} |x|^2  |Q_0(x)|^{2+4/d} \d x. 
\end{align*}
Next, we use \eqref{eq:iii-0} with $u=v_\ell/\|v_\ell\|_{L^2}$. Using $\|v_\ell\|_{L^2}\le 1$ and the above computations, together with the important identity \eqref{eq:GS-Q0},  we get
\begin{align*}
E_k &\le  \frac{1}{\|v_\ell\|_{L^2}^2}  \int_{\R^d} |\nabla v_\ell (x)|^2 \d x + \frac{1}{\|v_\ell\|_{L^2}^4} \int_{\R^d}  (\eps |x|^2-a^*)  |v_\ell(x)|^{2+4/d} \d x\\
&\le \frac{1}{\|v_\ell\|_{L^2}^4}  \left(  \int_{\R^d} |\nabla v_\ell (x)|^2 \d x + \int_{\R^d}  (\eps |x|^2-a^*)  |v_\ell(x)|^{2+4/d} \d x \right) \\
&\le (1+o(1)_{\ell\to \infty}) \left( \eps \int_{\R^d} |x|^2  |Q_0(x)|^{2+4/d} \d x + o(1)_{\ell\to \infty} \right).  
\end{align*}
Since $E_k$ is independent of $\ell$, by taking $\ell\to \infty$ we obtain
\begin{align*}
E_k \le \eps \int_{\R^d} |x|^2  |Q_0(x)|^{2+4/d} \d x.  
\end{align*}
Since it holds for arbitrary $\eps>0$, by taking $\eps\to 0$ we conclude that $E_k\le 0$. Thus $E_k=0$. 

Finally, if $E_k$ has a minimizer $u_0$, then using $\inf k \ge -a^* $ and the Gagliardo-Nirenberg inequality \eqref{eq:GN} we have
$$
0= \cE_k(u_0) \ge \int_{\R^d} |\nabla u_0|^2 - a^* \int_{\R^d} |u_0|^{2+4/d} \ge 0
$$
which implies that $u_0$ is an optimizer of \eqref{eq:GN}. This means $u_0$ is equal to $Q_0=Q/\|Q\|_{L^2}$ up to translations and dilations, and in particular $|u_0(x)|>0$ for all $x$. On the other hand, by  the Gagliardo-Nirenberg inequality \eqref{eq:GN}  again, we have
$$
0= \cE_k(u_0) \ge \int_{\R^d} (k(x)+a^*) |u_0|^{2+4/d} \ge 0.
$$
Since $k(x)+a^*\ge 0$ and $|u_0(x)|>0$ for all $x$, we conclude that $k(x)+a^*=0$ for all $x$, namely $k\equiv -a^*$.

\begin{remark} \label{rmk:local-int} In the critical case $\inf k=-a^*$, if $(k(x)+a^*)^{-d/2}$ is integrable, then by using H\"older's inequality as in the proof of (i), i.e. 
\begin{align*}
\int_{\R^d} |u|^2 &\le  \left( \int_{\R^d} (k(x)+a^*) |u(x)|^{2+4/d} \right)^{d/(d+2)}  \left( \int_{\R^d} \frac{1}{(k(x)+a^*)^{d/2}} \right)^{2/(d+2)} \\
&\le (\cE_k(u))^{d/(d+2)}  \left( \int_{\R^d} \frac{1}{(k(x)+a^*)^{d/2}} \right)^{2/(d+2)},
\end{align*}
we obtain $E_k>0$. The degeneracy condition \eqref{eq:local-degenerate} basically rules out the local integrability of $(k(x)+a^*)^{-d/2}$, and hence it is important to ensures that $E_k=0$. 
\end{remark}

\bigskip

(iv) Now assume that $\inf k<-a^*$. Since the function $k$ is continuous, there exist $\eps>0$ and a ball $B(x_0,r)\subset \R^d$ such that 
$$k(x)\le -a^*-\eps, \quad \forall x\in B(x_0,r).$$
By the variational principle, we have
\begin{align} \label{eq:iv-instability}
E_k &\le \cE_k(u(.-x_0)) = \int_{\R^d} |\nabla u(x)|^2 \d x+  \int_{\R^d} k(x+x_0)|u(x)|^{2+4/d} \d x  \nn\\
&  \le \int_{\R^d} |\nabla u|^2 - (a^*+\eps) \int_{\R^d} |u|^{2+4/d}
\end{align}
for all $u\in H^1(\R^d)$ supported on $B(0,r)$ such that $\|u\|_{L^2}=1$. 

We will use \eqref{eq:iv-instability} with suitable trial functions $u$. Let $Q_0=Q/\|Q\|_{L^2}$ be the (normalized) optimizer of the Gagliardo-Nirenberg inequality  \eqref{eq:GN}. Since  
$$
\int_{\R^d} |\nabla Q_0|^2 - a^* \int_{\R^d} |Q_0|^{2+4/d}=0
$$
and $C_c^\infty(\R^d)$ is dense in $H^1(\R^d)$, by approximating $Q_0$ we can find a function $\varphi\in C_c^\infty(\R^d)$ such that $\|\varphi\|_{L^2}=1$ and 
\bq \label{eq:varphi-a-eps}
\int_{\R^d} |\nabla \varphi|^2 - (a^*+\eps) \int_{\R^d} |\varphi|^{2+4/d}<0. 
\eq
Next, for any $\ell>0$ define 
$$
u_\ell(x)= \ell^{d/2} \varphi(\ell(x-x_0)).  
$$
Since $\varphi$ has compact support, if $\ell>0$ is sufficiently large, then $u_\ell$ is supported on $B(x_0,r)$. Thus we can use \eqref{eq:iv-instability} with the trial function $u_\ell$, which gives
\begin{align*}
E_k &\le  \int_{\R^d} |\nabla u_\ell|^2 - (a^*+\eps) \int_{\R^d} |u_\ell|^{2+4/d}\\
&= \ell^2 \left( \int_{\R^d} |\nabla \varphi|^2 - (a^*+\eps) \int_{\R^d} |\varphi|^{2+4/d} \right)
\end{align*}
for all $\ell>0$ sufficiently large. Taking $\ell\to \infty$ and using \eqref{eq:varphi-a-eps} we conclude that $E_k=-\infty$. 
\end{proof}

\section{Blow-up analysis} \label{sec:blowup}

\begin{proof}[Proof of Theorem \ref{thm:blowup}] {\bf Step 1: Energy upper bound.} This is done similarly as in the proof of Theorem \ref{thm:existence}. Without loss of generality let us assume that $x_1\in Z$. Then for any $\eps>0$, there exists $r=r_\eps>0$ such that 
$$
k(x) \le (\lambda +\eps)|x-x_1|^{p}-a, \quad \forall x\in B(x_1,r).    
$$
Then by the variational principle, for any $u\in H^1(\R^d)$, supported on $B(0,r)$ with $\|u\|_{L^2}=1$ we have
\bq \label{eq:Ek-upper-0}
E_k \le \cE_k(u(.-x_1)) \le \int_{\R^d} |\nabla u(x)|^2 + \int_{\R^d} ((\lambda +\eps)|x|^{p}-a)|u(x)|^{2+4/d}\d x.  
\eq

Now we choose a trial function $u$. Let $Q_0=Q/\|Q\|_{L^2}$ be the (normalized) optimizer of the Gagliardo-Nirenberg inequality  \eqref{eq:GN}. Take a smooth function $w:\R^d \to [0,1]$ such that $w(x)=1$ for $|x|\le r/2$ and $w (x)=0$ if $|x|\ge r$. In \eqref{eq:Ek-upper-0} we choose
$$
u= v_\ell/\|v_\ell\|_{L^2}, \quad \text{with} \quad  v_\ell(x) = \ell^{d/2}Q_0(\ell x) w(x).
$$
Since both $Q_0$ and $|\nabla Q_0|$ are exponentially decay (see \cite[Proposition 4.1]{GidNiNir-81}), we have 
\begin{align*}
\int_{\R^d} |v_\ell(x)|^2 \d x&=  \int_{\R^d} | Q_0|^2 + o(\ell^{-\infty})_{\ell\to \infty}  ,\\
\int_{\R^d} |\nabla v_\ell(x)|^2 \d x&= \ell^2 \int_{\R^d} |\nabla Q_0|^2 + o(\ell^{-\infty})_{\ell\to \infty},\\
\int_{\R^d} |v_\ell(x)|^{2+4/d} \d x &=\ell^2 \int_{\R^d} | Q_0|^{2+4/d}+ o(\ell^{-\infty})_{\ell\to \infty},\\
\int_{\R^d} |x|^p |v_\ell(x)|^{2+4/d} \d x & \le \ell^{2-p} \int_{\R^d} |x|^p  |Q_0(x)|^{2+4/d} \d x. 
\end{align*}
Here $o(\ell^{-\infty})_{\ell\to \infty}$ means an error smaller than any polynomial decay $o(\ell^{-s})$ with $s>0$. Combining with  $\|v_\ell\|_{L^2}\le 1$ and \eqref{eq:GS-Q0},  we get from \eqref{eq:Ek-upper-0}  that 
\begin{align*}
E_k &\le  \frac{1}{\|v_\ell\|_{L^2}^2}  \int_{\R^d} |\nabla v_\ell (x)|^2 \d x + \frac{1}{\|v_\ell\|_{L^2}^4} \int_{\R^d}  ((\lambda+\eps) |x|^p-a)  |v_\ell(x)|^{2+4/d} \d x\\
&\le \frac{1}{\|v_\ell\|_{L^2}^4}  \left(  \int_{\R^d} |\nabla v_\ell (x)|^2 \d x + \int_{\R^d}  ((\lambda+\eps) |x|^p-a)  |v_\ell(x)|^{2+4/d} \d x \right) \\
&\le (1+o(\ell^{-\infty})_{\ell\to \infty}) \times \\
&\times \left( \ell^2 (a^*-a) \int_{\R^d} |Q_0|^{2+4/d} + \ell^{2-p} (\lambda+\eps) \int_{\R^d} |x|^p  |Q_0|^{2+4/d} +  o(\ell^{-\infty})_{\ell\to \infty} \right).
\end{align*}
Choosing 
$$
\ell = (a^*-a)^{-1/p}  \xi
$$
with a constant $\xi>0$ independent of $a$, we obtain
$$
\limsup_{a\uparrow a^*} \frac{E_k}{(a^*-a)^{1-2/p}} \le \xi^2 \int_{\R^d} |Q_0|^{2+4/d} +  \xi^{2-p} (\lambda+\eps) \int_{\R^d} |x|^p  |Q_0|^{2+4/d}.  
$$
Since $\xi>0$ and $\eps>0$ can be chosen arbitrarily, we conclude that 
\begin{align*}
\limsup_{a\uparrow a^*} \frac{E_k}{(a^*-a)^{1-2/p}} \le \inf_{\xi>0} \left[ \xi^2 \int_{\R^d} |Q_0|^{2+4/d} +  \xi^{2-p} \lambda \int_{\R^d} |x|^p  |Q_0|^{2+4/d} \right].
\end{align*}
The right side is attained its minimum value at
$$
\xi_0= \left( \frac{(p-2)\lambda \int_{\R^d} |x|^p  |Q_0|^{2+4/d} }{2\int_{\R^d} |Q_0|^{2+4/d}} \right)^{1/p}= \left(\frac{(p-2)\lambda \int_{\R^d} |x|^p  |Q|^{2+4/d} }{2\int_{\R^d} |Q|^{2+4/d}} \right)^{1/p},
$$
and hence
\bq \label{eq:Ek-upper-sharp}
\limsup_{a\uparrow a^*} \frac{E_k}{(a^*-a)^{1-2/p}} \le  \left( \frac{\lambda p}{2} \int_{\R^d} |x|^p |Q_0|^{2+4/d} \right)^{2/p} \left( \frac{p}{p-2} \int_{\R^d} |Q_0|^{2+4/d}\right)^{1-2/p}.
\eq
\bigskip

\noindent
{\bf Step 2: Kinetic energy estimates.} Now take $u_a$ be a minimizer for $E_k$ with $k=K-a$. {When $a$ is sufficiently close $a^*$}, let us prove that
\bq \label{eq:kinetic-round-bound}
C (a^*-a)^{-2/p}  \ge \int_{\R^d} |\nabla u_a|^2 \ge a^* \int_{\R^d} |u_a|^{2+4/d} \ge C^{-1}(a^*-a)^{-2/p}
\eq
for a constant $C>0$ independent of $a$. 

Using the  Gagliardo-Nirenberg inequality \eqref{eq:GN} we have
\begin{align} \label{eq:Ek>=kin+K}
E_k&= \cE_k(u_a)= \int_{\R^d} |\nabla u_a|^2 - a \int_{\R^d} |u_a|^{2+4/d} + \int_{\R^d} K |u_a|^{2+4/d}\nn \\
&\ge \left(1-\frac{a}{a^*}\right)  \int_{\R^d} |\nabla u_a|^2 +  \int_{\R^d} K |u_a|^{2+4/d}. 
\end{align}
Since $K\ge 0$ we have
$$
E_k \ge \left(1-\frac{a}{a^*}\right)  \int_{\R^d} |\nabla u_a|^2 
$$
and the first inequality in \eqref{eq:kinetic-round-bound} follows from the energy upper bound \eqref{eq:Ek-upper-sharp}. The second inequality in \eqref{eq:kinetic-round-bound}  is exactly the  Gagliardo-Nirenberg inequality  \eqref{eq:GN}. The most difficult part is the third inequality in  \eqref{eq:kinetic-round-bound}. Inspired by Guo-Seiringer \cite{GuoSei-14} (see also \cite{Phan-17a}), we will prove that a substantial part of the mass of $u_a$ concentrates close to the minima $x_j$ of $K$. However, the perturbation method in  \cite{GuoSei-14,Phan-17a} does not work in our case and we have to develop new ideas in the proof below. 

Our key point is to use the fact that $K(x)^{-d/2}$ is integrable away from the minima $x_j$. To be precise, for $r>0$ small let us denote
$$
A_r= \{x\in \R^d: |x-x_j|\ge r \text{ for all }j=1,2,...,J\}.
$$ 
From the assumption on $K$, we know that if $0<r<R$ small, then 
$$
K(x)\le C \max_j |x-x_j|^p, \quad \forall x\in A_r \setminus A_R
$$
and $K(x)^{-d/2}$ is  integrable on $A_R$. We will take $R$ small but fixed (independent of $a$) and choose $r=r_a$ small. Consequently, 
\begin{align} \label{eq:int-K-Ar}
\int_{A_r} K(x)^{-d/2} \d x &= \int_{A_r \setminus A_R} K(x)^{-d/2} \d x + \int_{A_R} K(x)^{-d/2} \d x \nn \\
&\le C \sum_j \int_{r\le |x-x_j|\le R} |x-x_j|^{-dp/2} \d x + C_R \nn\\
&\le C r^{d(1-p/2)}
\end{align}
for a constant $C>0$ independent of $r$.

Now we estimate the mass of $u_a$ away from minima $\{x_j\}$ using Holder's inequality
\bq \label{eq:xddsydsyd}
\int_{A_r} |u_a|^2 \le \left( \int_{A_r} K|u_a|^{2+4/d} \right)^{d/(d+2)} \left( \int_{A_r} K^{-d/2} \right)^{2/(d+2)}. 
\eq
From \eqref{eq:Ek>=kin+K} and the upper bound on $E_k$ in \eqref{eq:Ek-upper-sharp} we have
$$
\int_{A_r} K|u_a|^{2+4/d} \le \int_{\R^d} K|u_a|^{2+4/d} \le E_k \le C (a^*-a)^{1-2/p}. 
$$
Combining with \eqref{eq:int-K-Ar} we obtain from \eqref{eq:xddsydsyd} that 
\begin{align*}
\int_{A_r} |u_a|^2 &\le \left( C (a^*-a)^{1-2/p}  \right)^{d/(d+2)} \left(C r^{d(1-p/2)} \right)^{2/(d+2)} \\
&= C \left[ \frac{r}{(a^*-a)^{1/p}} \right]^{\frac{d}{d+2}(2-p)}.
\end{align*}
Since $p>2$, if we choose 
\bq \label{eq:r-choice}
r=C_0(a^*-a)^{1/p}
\eq
with a big, fixed constant $C_0>0$, then we conclude that
\begin{align*}
\int_{A_r} |u_a|^2 \le \frac{1}{2}.
\end{align*}
Since $\|u_a\|_{L^2}=1$, the latter bound is equivalent to
\begin{align} \label{eq:local-ua}
\int_{\R^d \setminus A_r} |u_a|^2 \ge \frac{1}{2}.
\end{align}

Finally, using H\"older's inequality again (with the above choice \eqref{eq:r-choice} of $r$) we have
\begin{align*}
\frac{1}{2}\le \int_{\R^d\setminus A_r} |u_a|^2 &\le \left( \int_{\R^d\setminus A_r} |u_a|^{2+4/d} \right)^{d/(d+2)} \left( \int_{\R^d\setminus A_r} 1 \right)^{2/(d+2)} \\
&\le \left( \int_{\R^d} |u_a|^{2+4/d} \right)^{d/(d+2)} (C r^{d})^{2/(d+2)},
\end{align*}
which implies the third inequality in \eqref{eq:kinetic-round-bound}
$$
\int_{\R^d} |u_a|^{2+4/d}  \ge C^{-1} (a^*-a)^{-2/p}. 
$$

\bigskip

\noindent
{\bf Step 3: Convergence of minimizers by compactness argument.} We recall the following well-known compactness results for the variational problem
\bq \label{eq:GP-0}
0= \inf\left\{ \int_{\R^d} |\nabla u|^2 - a^* \int_{\R^d} |u|^{2+4/d}: u\in H^1(\R^d), \|u\|_{L^2}=1 \right\}. 
\eq

\begin{lemma} \label{lem:compactness-GN}Let $\{\varphi_n\}$ be a minimizing sequence for  the variational problem \eqref{eq:GP-0} such that 
$$
C^{-1} \le \int_{\R^d} |\nabla \varphi_n|^2 \le C
$$ 
for a constant $C>0$ independent of $n$. Then up to a subsequence when $n\to \infty$, there exist {$\theta\in \R$}, $b>0$ and $\{z_n\} \subset \R^d$ such that 
$$
\varphi_n (x-z_n) \to e^{i\theta} b^{d/2} Q_0(bx)
$$
strongly in $H^1(\R^d)$. 
\end{lemma}

This lemma follows from the standard concentration-compactness method \cite{Lions-84,Lions-84b} (see e.g. \cite[Appendix A]{Phan-18} for a detailed  explanation).

To apply Lemma \ref{lem:compactness-GN}, we need to rescale $u_a$ to ensures that its kinetic energy is of order 1. 
Denote
$$
v_a(x)= (a^*-a)^{d/(2p)} u_a((a^*-a)^{1/p} x), 
$$
i.e.
$$u_a(x) =   (a^*-a)^{-d/(2p)} v_a((a^*-a)^{-1/p} x).$$
Using $K\ge 0$ we obtain
\begin{align*}
E_k &\ge \int_{\R^d} |\nabla u_a|^2 - a \int_{\R^d} |u_a|^{2+4/d}\\
&=(a^*-a)^{-2/p} \left[  \int_{\R^d} |\nabla v_a|^2 - a  \int_{\R^d} |v_a|^{2+4/d} \right]. 
\end{align*}
Combining with the upper bound on $E_k$ in \eqref{eq:Ek-upper-sharp}, we find that 
$$
 \int_{\R^d} |\nabla v_a|^2 - a  \int_{\R^d} |v_a|^{2+4/d} \le C (a^*-a) \to 0.
$$
Thus $\{v_a\}$ is a minimizing sequence for  the variational problem \eqref{eq:GP-0} as $a\uparrow a^*$. Moreover, from the kinetic estimate \eqref{eq:kinetic-round-bound} from Step 2, we find that 
$$
C^{-1} \le \int_{\R^d} |\nabla v_a|^2 \le C
$$
for a constant $C>0$ independent of $a$. Thus by Lemma  \ref{lem:compactness-GN}, up to a subsequence (i.e. $a_n\uparrow a^*$, but we will write $a\uparrow a^*$ for simplicity) and up to a phase, there exist a constant $b>0$ and a sequence $\{z_a\} \subset  \R^d$ such that
\bq \label{eq:CV-va}
v_a(x-z_a)\to b^{d/2} Q_0(bx)
\eq
strongly in $H^1(\R^d)$ as $a\uparrow a^*$. 

\bigskip

\noindent
{\bf Step 4: Determination of $\{z_a\}$.} Now let us give more information on the sequence $\{z_a\}$ in \eqref{eq:CV-va}. Recall  \eqref{eq:local-ua} which implies that 
$$
\sum_{j} \int_{|x-x_j|\le r} |u_a|^2 \ge \frac{1}{2}. 
$$
Recalling the choice \eqref{eq:r-choice} of $r$ and the definition of $v_a$, we obtain, by the change of variable {$y=(a^*-a)^{-1/p} x +z_a$},
$$
\sum_{j} \int_{|y-z_a - (a^*-a)^{-1/p}x_j|\le C_0} {|v_a(y-z_a)|^2} \ge \frac{1}{2}. 
$$
From the strong convergence \eqref{eq:CV-va}, we deduce that
$$
\sum_{j} \int_{|y -z_a - (a^*-a)^{-1/p}x_j|\le C_0} b^{d}|Q_0(by)|^2 \ge \frac{1}{2}. 
$$
Thus we can find some $j_0\in \{1,2,...,J\}$ such that
$$
\int_{|y -z_a - (a^*-a)^{-1/p}x_{j_0}|\le C_0} b^{d}|Q_0(by)|^2 \ge \frac{1}{2J}>0
$$
along a subsequence $a\uparrow a^*$. Since $Q_0$ exponentially decays, the latter bound implies that 
{$z_a + (a^*-a)^{-1/p}x_{j_0}$} is bounded. Thus up to a subsequence again, we can find $x_0\in \R^d$ such that  
{$$
z_a+(a^*-a)^{-1/p}x_{j_0} \to x_0.
$$}
Since the translation action is continuous in $L^2(\R^d)$, we can eventually replace $z_a$ by {$x_0-(a^*-a)^{-1/p}x_{j_0} $} in \eqref{eq:CV-va} and obtain
\bq \label{eq:CV-vab}
{v_a(x+(a^*-a)^{-1/p}x_{j_0}-x_0)}\to b^{d/2} Q_0(bx)
\eq
strongly in $H^1(\R^d)$ as $a\uparrow a^*$. 

In the following we will prove that $x_{j_0}\in Z$ (the set of flattest minima),  $x_0=0$ and determine $b$ exactly. All this requires an exact asymptotic analysis of the energy $E_k$. The sharp upper bound has been given in Step 1, and now we focus on the matching lower bound.  

\bigskip

\noindent 
{\bf Step 5: Energy lower bound.} Using the Gagliardo-Nirenberg inequality \eqref{eq:GN} as in  \eqref{eq:Ek>=kin+K} and putting back the definition of $v_a$, we have
\begin{align} \label{eq:Ek-low-s-0}
E_k =\cE_k(u_a) &\ge (a^*-a)\int_{\R^d} |u_a|^{2+4/d} + \int_{\R^d} K |u_a|^{2+4/d} \nn \\
&=(a^*-a)^{1-2/p} \int_{\R^d} |v_a|^{2+4/d} \nn\\
&\quad +  (a^*-a)^{-2/p} \int_{\R^d} K(x_{j_0}+(a^*-a)^{1/p}x) \times \\
&\quad\quad\quad\quad\quad\quad\quad\quad\quad\quad\times  {|v_a(x+(a^*-a)^{-1/p}x_{j_0})|}^{2+4/d} \d x \nn.
\end{align}

The first term on the right side of \eqref{eq:Ek-low-s-0} can be estimated exactly using \eqref{eq:CV-vab} and Sobolev's embedding theorem,  i.e.
$$
\int_{\R^d} |v_a|^{2+4/d} \to  \int_{\R^d} |b^{d/2} Q_0(bx)|^{2+4/d} \d x = b^2 \int_{\R^d} |Q_0|^{2+4/d}. 
$$

To deal with the second term on the right side of \eqref{eq:Ek-low-s-0}, let us use the local information of  $K$ around its minima $x_{j_0}$:
$$
\lim_{x\to x_{j_0}} {\frac{K(x)}{|x-x_{j_0}|^{p_{j_0}}}} = \lambda_{j_0}>0,
$$
or putting differently, 
$$
\lim_{a\uparrow a^*} \frac{K(x_{j_0}+(a^*-a)^{1/p}x)}{((a^*-a)^{1/p}|x|)^{p_{j_0}} } = \lambda_{j_0}, \text{for all } x\in \R^d.
$$
Moreover, the convergence \eqref{eq:CV-vab} implies that, up to a subsequence, 
$$
v_a({x+(a^*-a)^{-1/p}x_{j_0}}) \to b^{d/2}Q_0(b(x+x_0))\quad \text{for a.e. } x\in \R^d.
$$
Thus we have the pointwise convergence
\begin{align*}
\lim_{a\uparrow a^*} & (a^*-a)^{-p_{j_0}/p} K(x_{j_0}+(a^*-a)^{1/p}x) |v_a({x+(a^*-a)^{-1/p}x_{j_0}})|^{2+4/d}\\
&= \lambda_{j_0} |x|^{p_{j_0}} \Big(b^{d/2}Q_0(b(x+x_0)) \Big)^{2+4/d}. 
\end{align*}
Therefore, we can estimate the second term on the right side of \eqref{eq:Ek-low-s-0} using Fatou's lemma 
\begin{align*}
\liminf_{a\uparrow a^*} &  (a^*-a)^{-p_{j_0}/p} \int_{\R^d} K(x_{j_0}+(a^*-a)^{1/p}x) |v_a({x+(a^*-a)^{-1/p}x_{j_0}})|^{2+4/d}\\
&\ge \int_{\R^d} \lambda_{j_0} |x|^{p_{j_0}} \Big(b^{d/2}Q_0(b(x+x_0)) \Big)^{2+4/d} \\
&=  \lambda_{j_0} b^{2-p_{j_0}} \int_{\R^d} |x|^{p_{j_0}} |Q_0(x+bx_0)|^{2+4/d} \d x. 
\end{align*}

In summary, we deduce from \eqref{eq:Ek-low-s-0} that
\begin{align*} 
E_k &\ge (a^*-a)^{1-2/p} \left[ b^2 \int_{\R^d}  |Q_0|^{2+4/d} + o(1)_{a\to a^*}\right] \nn\\
& + (a^*-a)^{p_{j_0}/p-2/p}  \left[ \lambda_{j_0} b^{2-p_{j_0}} \int_{\R^d} |x|^{p_{j_0}} |Q_0(x+bx_0)|^{2+4/d} \d x + o(1)_{a\to a^*}\right]. 
\end{align*}

Here we know by a-priori that $p_{j_0}\le p$. However, if $p_{j_0}<p$, then 
$$ (a^*-a)^{p_{j_0}/p-2/p} \gg (a^*-a)^{1-2/p}$$
in the limit $a\uparrow a^*$, leading to a contradiction to the upper bound \eqref{eq:Ek-upper-sharp} in Step 1. Thus we must have
$$
p_{j_0}=p
$$
and hence
\begin{align}  \label{eq:Ek-low-sh-1}
\liminf_{a\uparrow a^*} \frac{E_k}{(a^*-a)^{1-2/p}} &\ge  b^2 \int_{\R^d}  |Q_0|^{2+4/d} \nn\\
&\quad  + \lambda_{j_0} b^{2-p} \int_{\R^d} |x|^{p} |Q_0(x+bx_0)|^{2+4/d} \d x .
\end{align} 

\bigskip

\noindent 
{\bf Step 6: Conclusion.} Combining the upper bound \eqref{eq:Ek-upper-sharp} and the lower bound \eqref{eq:Ek-low-sh-1} we find that
\begin{align} \label{eq:matching}
& b^2 \int_{\R^d}  |Q_0|^{2+4/d}+ \lambda_{j_0} b^{2-p} \int_{\R^d} |x|^{p} |Q_0(x+bx_0)|^{2+4/d} \d x \nn\\
&\le \inf_{\xi>0} \left[ \xi^2 \int_{\R^d} |Q_0|^{2+4/d} +  \lambda \xi^{2-p}  \int_{\R^d} |x|^p  |Q_0(x)|^{2+4/d} \d x \right].
\end{align}
Since $p_{j_0}=p$, by the definition of $\lambda$ we have 
$$\lambda_{j_0} \ge \lambda$$
where the equality happens if and only if $x_{j_0}\in Z$ (the set of flattest minima). Moreover, since $Q_0$ is radially symmetric decreasing and $|x|^p$ is radially symmetric (strictly) increasing, by the rearrangement inequality {(see \cite[Theorem 3.4 and the associated remark]{LiLo-01})} we have
$$
\int_{\R^d} |x|^p  |Q_0(x+bx_0)|^{2+4/d} \le \int_{\R^d} |x|^p  |Q_0(x)|^{2+4/d} \d x
$$
with the equality happens if and only if $x_0=0$. From the matching identity \eqref{eq:matching} we conclude that
$$
x_{j_0}\in Z, \quad x_0=0
$$
and $b$ is exactly the optimal value on the right side of \eqref{eq:matching}
$$
b=\xi_0= \left(\frac{(p-2)\lambda \int_{\R^d} |x|^p  |Q|^{2+4/d} }{2\int_{\R^d} |Q|^{2+4/d}} \right)^{1/p}. 
$$
Thus we have proved the desired energy convergence from \eqref{eq:Ek-upper-sharp}-\eqref{eq:Ek-low-sh-1} 
$$
\liminf_{a\uparrow a^*} \frac{E_k}{(a^*-a)^{1-2/p}} = \inf_{\xi>0} \left[ \xi^2 \int_{\R^d} |Q_0|^{2+4/d} +  \lambda \xi^{2-p}  \int_{\R^d} |x|^p  |Q_0(x)|^{2+4/d} \d x \right]
$$
and the ground state convergence from \eqref{eq:CV-vab}
$$
v_a({x+ (a^*-a)^{-1/p}x_{j_0}}) \to b^{d/2} Q_0 (b x).
$$
strongly in $H^1(\R^d)$, which is equivalent to \eqref{eq:GS-cv}:
$$
(a^*-a)^{d/(2p)} u_a(x_{j_0}+ (a^*-a)^{1/p}x) \to b^{d/2} Q_0 (b x).
$$

So far, we have to prove these convergences up to a subsequence $a_n \uparrow a^*$. However, since the limit in the energy convergence is unique, the energy convergence holds for the whole family $a\uparrow a^*$. Moreover, if $Z$ has a unique element, then the limit of the ground state convergence is also unique, and the ground state convergence holds for the whole family $a\uparrow a^*$. This ends the proof. 
\end{proof}

\end{document}